\documentclass[preprint, 12pt]{elsarticle}

\usepackage{amsthm,amsmath,amssymb,url}
\usepackage[classfont=sanserif,langfont=caps,funcfont=italic]{complexity}

\newclass{\Wone}{W[1]}
\newclass{\paraNP}{para\text{-}NP}

\newlang{\WEDCE}{WEDCE}
\newlang{\WDCE}{WDCE}
\newlang{\WERE}{WERE}
\newlang{\WSRE}{WSRE}
\newlang{\WBDE}{WBDE}

\newcommand{\iWDCE}{\ensuremath{_{\infty}}\WDCE{}}

\newcounter{reductionCounter}
\newtheorem{theorem}{Theorem}
\newtheorem{lemma}[theorem]{Lemma}
\newtheorem{proposition}[theorem]{Proposition}
\newtheorem{corollary}[theorem]{Corollary}
\newtheorem{claim}[theorem]{Claim}

\theoremstyle{definition}
\newtheorem{definition}[theorem]{Definition}

\newtheorem{redrule}[reductionCounter]{Reduction Rule}

\newcommand{\Card}[1]{\left|#1\right|}
\newcommand{\vdel}{\mathsf{v}}
\newcommand{\edel}{\mathsf{e}}
\newcommand{\eadd}{\mathsf{a}}
\newcommand{\Yes}{\textsc{Yes}}
\newcommand{\No}{\textsc{No}}

\newcommand{\plain}{1}

\newcommand{\problem}[3]{\begin{quote}\noindent\textsc{#1}:\\
                                    \begin{tabular}{rp{0.6\textwidth}}
                                      \textit{Instance: } &#2\\
                                      \textit{Question: } &#3
                                    \end{tabular}\end{quote}}

\DeclareMathOperator{\tr}{tr}
\DeclareMathOperator{\tw}{tw}

\journal{Theoretical Computer Science}

\begin{document}

\begin{frontmatter}
\title{Graph Editing Problems with Extended Regularity Constraints}
\author{Luke Mathieson}
\address{Centre for Information Based Medicine, Bioinformatics and Biomarker Discovery, University of Newcastle, Australia}
\begin{abstract}
  Graph editing problems offer an interesting perspective on sub- and supergraph identification problems for a large variety of target properties. They have also attracted significant attention in recent years, particularly in the area of parameterized complexity as the problems have rich parameter ecologies.

In this paper we examine generalisations of the notion of editing a graph to obtain a regular subgraph. In particular we extend the notion of regularity to include two variants of edge-regularity along with the unifying constraint of strong regularity. We present a number of results, with the central observation that these problems retain the general complexity profile of their regularity-based inspiration: when the number of edits $k$ and the maximum degree $r$ are taken together as a combined parameter, the problems are tractable (i.e. in \FPT{}), but are otherwise intractable.

We also examine variants of the basic editing to obtain a regular subgraph problem from the perspective of parameterizing by the treewidth of the input graph. In this case the treewidth of the input graph essentially becomes a limiting parameter on the natural $k+r$ parameterization.
\end{abstract}

\begin{keyword}
graph algorithms \sep computational complexity \sep algorithms \sep parameterized complexity \sep graph editing
\end{keyword}

\end{frontmatter}

\section{Introduction}
\label{sec:intro}

Graph editing problems --- problems where \emph{editing operations} are applied to an input graph to obtain a graph with a given property --- provide an interesting and flexible framework for considering many graph problems. For example virtually any problem whose witness structure is a subset of the vertices of the input graph can be alternatively phrased as a problem regarding the deletion of vertices to obtain a suitable property. \textsc{$k$-Vertex Cover} can be viewed as the problem of deleting at most $k$ vertices such that the resultant graph has no edges. This target property can be viewed as a \emph{degree contraint}; the degree of all vertices in the final graph should be zero. The nature of the resulting graph can also be defined by which editing operations are allowed; vertex deletion alone results in induced subgraphs, edge deletion alone produces spanning subgraphs, and so on. 

Degree contraint editing problems have long been of interest to comuputational complexity theorists and this interest has been echoed in the parameterized complexity context. The proof of the \NP{}-completeness of \textsc{Cubic Subgraph} was attributed to Chv\'{a}tal by Garey and Johnson~\cite{GJ79}. This naturally generalises to \textsc{$r$-Regular Subgraph}, which we can consider as the problem of removing vertices and edges to obtain an $r$-regular graph. \textsc{$r$-Regular Subgraph} is \NP{}-complete~\cite{Plesnik84} for $r \geq 3$, even under a number input constraints~\cite{CC90,S94,S96,S97}. The problem of finding a maximum \emph{induced} $r$-regular subgraph is \NP{}-complete for $r \geq 0$~\cite{CKL07} ($r=0$ is \textsc{Maximum Independent Set} and the removed vertices form a minimum vertex cover). If we allow only edge deletion, we have the \textsc{$r$-Factor} problem. For $r = 1$ this is the basic matching problem, well known to be polynomial~\cite{Edmonds65a,Edmonds65b,Kuhn55,Kuhn56}. If $r > 1$, or indeed if each vertex has a different target degree (the \textsc{$f$-Factor} problem), Tutte gives a reduction a polynomial-time solvable matching problem~\cite{Tutte54,Tutte74}. This problem can be further generalised by giving each vertex a range of target degrees (the \textsc{Degree Constrained Subgraph} problem) and is polynomial-time solvable~\cite{UrquhartPhD67}. If the edges of the graph have capacities (the \textsc{(Perfect) $b$-Matching} problem), the problem remains in \P{}, using Tutte's \textsc{$f$-Factor} algorithm~\cite{KorteVygen08}. In the \textsc{General Factor}~\cite{Lovasz70,Lovasz72} problem allows each vertex to have a list of target degrees. If the lists contain gaps of greater than 1, the problem is \NP{}-complete and polynomial-time solvable otherwise~\cite{Cornuejols88}. The problem of adding at most 2 vertices and a minimum number of edges to obtain a $\Delta$-regular supergraph, where $\Delta$ is the maximum degree of the input graph, is polynomial-time solvable~\cite{BLT03}.

In the parameterized complexity setting, deleting $k$ vertices to obtain an $r$ regular graph is \Wone{}-hard for $r \geq 0$ with parameter $k$~\cite{MathiesonSzeider12}, but \FPT{} with parameter $k+r$~\cite{MT08}. Mathieson and Szeider~\cite{MathiesonSzeider12} give a series of results for similar problems, which is extended further in Mathieson's doctoral thesis~\cite{Mathieson10}. In this they examine the \textsc{Weighted Degree Constrained Editing} (\WDCE{}) problem where the vertices and edges are weighted and each vertex has a set of target degrees, for combinations of vertex deletion, edge deletion and edge addition. When parameterized by the number $k$ of edits allowed, the problem is \Wone{}-hard, when parameterized by the number of edits and the maximum value $r$ in any of the degree lists, the problem is \FPT{}. When the target degree sets are singletons and the editing operations include only edge deletion and addition, the problem is in \P{}. The problem remains \Wone{}-hard with parameter $k$ even in the unweighted case where each vertex has the same target degree $r$. In the singleton case, where vertex deletion or vertex deletion and edge deletion is allowed, the problem has a kernel of size polynomial in $k+r$ (and hence is in \FPT{} for parameter $k+r$). The more general \WDCE{} problem with vertex deletion and/or edge deletion has a kernel of size exponential in $k+r$, and hence is \FPT{} for parameter $k+r$. Froese \emph{et al.}~\cite{FroeseNichterleinNiedermeier14} prove that no polynomial kernel is possible in these cases unless $\NP{} \subseteq \coNP{}/poly$. For the general weighted case with degree lists where the editing operations are any combination of vertex deletion, edge deletion and edge addition Mathieson and Szeider~\cite{MathiesonSzeider12} give a logic based proof of \FPT{} membership. Mathieson~\cite{Mathieson10} shows that if vertex deletion and edge addition are allowed (and perhaps edge deletion), then in the weighted case (even with singleton vertex lists), no polynomial kernel is possible unless $\NP{} \subseteq \coNP{}/poly$. In particular they give the following central theorem:

\begin{theorem}[Mathieson and Szeider~\cite{MathiesonSzeider12} Theorem 1.1]\label{thm:queen}
  For all non-empty subsets $S$ of $\{\vdel,\edel,\eadd\}$ the problem
  $\WDCE(S)$ is fixed-parameter tractable for parameter $k+r$, and
  $\W[1]$-hard for parameter $k$. If $\vdel\in S$ then $WDCE(S)$ remains
  $\W[1]$-hard for parameter $k$ even when all degree lists are restricted
  to~$\{r\}$ and all vertices and edges have unit weight~$1$.
\end{theorem}

Golovach~\cite{Golovach15} gives a concrete \FPT{} algorithm for the unweighted case with parameter $k+r$ where vertex deletion and edge addition are allowed and shows that this case has no polynomial kernel unless $\NP{} \subseteq \coNP{}/poly$. Froese \emph{et al.}, in addition to the results mentioned above, show that the unweighted case with degree lists and edge addition has a kernel of size polynomial in $k+r$. In fact they show that either the instance is polynomial-time solvable or the kernel is polynomially-sized in $r$ alone. Golovach~\cite{Golovach14b} looks at the case where the target graph must also remain connected. Dabrowski \emph{et al.}~\cite{DabrowskiGHPT15} look at the case where the input is planar and vertex deletion and edge deletion are allowed, and show that although still \NP{}-complete, a kernel polynomially-sized in the number of deletions is obtainable. Belmonte \emph{et al.}~\cite{BelmonteGvHP14} study the problem of using edge contraction to fulfil degree constraints.

More recently, Bulian and Dawar~\cite{Bulian2016} demonstrated a powerful meta-theorem based approach showing (amongst other results) that for a large collection of classes $\mathcal{C}$ of sparse graphs, determining the edit distance of an input graph $G$ from some graph in $\mathcal{C}$ is fixed-parameter tractable. This meta-theorem gives many (if not all) of the results discussed here as a corollary, with the obvious caveat being essentially non-constructive in algorithmic terms.

\subsection{Our Contribution}
\label{sec:contrib}

In this paper we look at problems with alternative forms of degree constraints: \emph{edge-degree-regularity}, \emph{edge-regularity} and \emph{strong-regularity} (q.v. Section~\ref{sec:def} for definitions). We show that for these constraints, and with any combination of vertex deletion, edge deletion and edge addition, these problems are typically fixed-parameter tractable with the combined parameter $k+r$, para-\NP{}-complete with parameter $r$ and \Wone{}-hard with parameter $k$.

We also consider the parameterization of certain \WDCE{} problems by the treewidth of the input graph where the number of editing operations is unbounded and show that finding an (induced) $r$-regular subgraph of graphs of bounded treewidth is in \FPT{}. When both vertex deletion and edge addition is allowed, the problem becomes trivially polynomial-time solvable by simply editing the graph into a $K_{r+1}$ if possible, and answering no otherwise.

\section{Definitions and Notation}
\label{sec:def}

We denote the closed (integer) interval from $a$ to $b$ by $[a,b]$. If $a=0$, we denote the interval $[0,b]$ by $[b]$. We denote the power set of a set $X$ by~$\mathcal{P}(X)$. 


In this paper we consider only simple, undirected graphs. Given a graph $G=(V,E)$ and two vertices $u$ and $v$ we denote the edge $\{u,v\}\in E$ by $uv$ or $vu$. The \emph{open neighbourhood} $N_{G}(u)$ of a vertex $u$ is the set $\{v \mid uv \in E\}$. The \emph{closed neighbourhood} $N_{G}[u]$ of a vertex $u$ is $N_{G}(u) \cup \{u\}$. The degree of a vertex $u$ is denoted $d(u)$ and $d(u) = \Card{N_{G}(u)}$. Given an edge $uv \in E$, the \emph{edge-degree} $d_{G}(uv)$ is the sum of the degrees of $u$ and $v$, i.e. $d_{G}(uv) = d_{G}(u) + d_{G}(v)$. A graph is $r$-\emph{regular} if for all $u \in V$ we have $d_{G}(u) = r$. If a graph is $r$-regular for some $r$, it is \emph{regular}. If for every edge $uv \in E$ we have $d_{G}(uv) = r$, we say $G$ is \emph{edge-degree}-$r$-\emph{regular}. A graph is $(r,\lambda)$-\emph{edge-regular} if every vertex has degree $r$ and every edge $uv$ has $\Card{N_{G}(u)\cap N_{G}(v)} = \lambda$. A graph is $(r,\lambda,\mu)$-\emph{strongly-regular} if it is $(r,\lambda)$-edge-regular and for every pair $u$, $v$ of non-adjacent vertices we have $\Card{N_{G}(u)\cap N_{G}(v)} = \mu$. The definitions are generalised to degree contraints in the problems we consider, see Section~\ref{sec:prob_def}.

We consider three graph editing operations: vertex deletion, edge deletion and edge addition. For brevity we denote these by $\vdel{}$, $\edel{}$ and $\eadd{}$ respectively. In a weighted graph, the cost of deleting a vertex with weight $w$ is $w$, the cost of deleting an edge with weight $w$ and the cost of adding an edge with weight $w$ is $w$. In an unweighted graph, the cost of each editing operation is~$1$.

Throughout this paper we consider a number of variants graph editing problems. In each case the input is a graph $G$ along with a weight function $\rho$ and a degree function $\delta$. We extend the normal notation for the degree of a vertex to the weighted degree of a vertex, denoted $d^{\rho}(v)$ where $d^{\rho}(v) = \sum_{u\in N_{G}(v)}\rho(uv)$. We also use this notation for the weighted edge-degree $d^{\rho}(uv)$.

For a given base editing problem $\Pi$, we denote by $\Pi_{1}$ the unweighted variant where $\rho(x) = 1$ for every $x$ in the domain of $\rho$, by $_{\infty}\Pi$ we denote the variant where the editing cost is removed (i.e. $k$ is no longer part of the problem), by $\Pi^{r}$ the variant where $\delta(x) = r$ (equivalently $\delta(x) = \{r\}$) for a fixed $r$ and all $x$ in the domain of $\delta$ and by $\Pi^{\ast}$ the variant where $\delta(x) = k_{x}$ (equivalently $\delta(x) = \{k_{x}\}$) for some $k_{x}$ for all $x$ in the domain of~$\delta$.

\subsection{Reduction Rules, Kernelization and Soundness}
\label{sec:reduct-rules-soundn}

A reduction rule is a self mapping from an instance of a problem to another instance of the same problem that reduces the instance size. In particular, in the Parameterized Complexity context, reduction rules typically form the basis of Kernelization algorithms. This technique is often called a ``reduction to problem kernel''. While we refer the reader to standard texts~\cite{DF99,DowneyF13,FG06,Niedermeier06} for the technical details of reduction rules and kernelization, we note here in particular the definition of \emph{soundness} of a reduction rule. A reduction rule is \emph{sound} if and only if it preserves yes instances and no instances, \emph{i.e.} it maps any yes instance to a yes instance and any no instance to a no instance of the given problem.

\subsection{Compositional Problems and Polynomial Sized Kernels}

For some problems we may suspect that they may not have kernels bounded by a polynomial in the parameter. Of course fixed-parameter tractability guarantees that they have a kernelization of some form, however this may also be impractical. Bodlaender \emph{et al.}~\cite{BodlaenderDFH09} develop a tool aimed at showing that problems do not have a polynomially sized kernel, based on some complexity theoretic assumptions.

\sloppypar\begin{definition}[Composition] A \emph{composition algorithm} for a parameterized problem $(\mathcal{P},\kappa)$ is an algorithm that receives as input a sequence $((x_{1},k), \ldots, (x_{t},k))$ of instances of $(\mathcal{P},\kappa)$ and outputs in time bounded by a polynomial in $\sum_{i=1}^{t} \Card{x_{i}}+k$ an instance $(x',k')$ where:
\begin{enumerate}
\item  $(x',k')$ is a \Yes{}-instance of $(\mathcal{P},\kappa)$ if and only if $(x_{i},k)$ is a \Yes{}-instance of $(\mathcal{P},\kappa)$ for some $i \in \{1,\ldots,t\}$.
\item $k' = p(k)$ where $p$ is a polynomial.
\end{enumerate}
\end{definition}

This definition is then accompanied by the key lemma:

\begin{lemma}[\cite{BodlaenderDFH09}]\label{lemma:comp_no_poly_kernel}
Let $(\mathcal{P},\kappa)$ be a parameterized problem with a composition algorithm where the non-parameterized version of the problem $\mathcal{P}$ is $\NP{}$-complete. Then if $(\mathcal{P},\kappa)$ has a polynomially sized kernel, the Polynomial Hierarchy collapses to the third level.
\end{lemma}

For details of the Polynomial Hierarchy we refer to Stockmeyer~\cite{Stockmeyer76}, and note that a collapse in the Polynomial Hierarchy seems unlikely~\cite{Papadimitriou94}.

Therefore any demonstration of the existence of a composition algorithm for a fixed-parameter tractable problem indicates that the problem is unlikely to have a polynomially sized kernel.

\subsection{Problem Definitions}
\label{sec:prob_def}

We now define the constraints of interest in this paper, and the consequent graph editing problems.

\subsubsection{Edge-Degree Regularity}\label{subsubsec:edge_degree_reg_def}

Edge-degree constraints naturally extend vertex based degree constraints, notably any $r$-regular graph is edge-degree $2r$-regular. However an edge-degree regular graph may not be regular. Therefore the class of edge-degree regular graphs forms a proper superclass of the class of regular graphs.

We define the \textsc{Weighted Edge Degree Constraint Editing} problem, or $\WEDCE$ similarly to the $\WDCE$ problem.

\problem{$\WEDCE{}(S)$}{A graph $G=(V,E)$, two integers $k$ and $r$, a weight function $\rho:V \cup E \rightarrow \{1,2,\dots\}$, and a degree list function $\delta: E \rightarrow \mathcal{P}([r])$.}{Can we obtain from $G$ a graph $G'=(V',E')$ using editing operations from $S$ only, such that for all $uv \in E'$ we have $d^{\rho}(uv) \in \delta(uv)$, with total editing cost at most $k$?}

We write $\WEDCE_\plain(S)$ to indicate that the given graph is unweighted, and we write $\WEDCE^*(S)$ if all degree lists are singletons; if all singletons are $\{r\}$ then we write $\WEDCE^{r}(S)$. We omit set braces whenever the context allows, and write, for example, $\WEDCE(\vdel)$ instead of $\WEDCE(\{\vdel\})$.







\subsubsection{Edge Regularity}\label{subsubsec:edge_reg_def}

A graph $G$ is $(r,\lambda)$-edge regular if every vertex has degree $r$ and every edge $uv \in E(G)$ has $\Card{N(u)\cap N(v)} = \lambda$. We define the \textsc{Weighted Edge Regularity Editing} ($\WERE{}$) problem.

\problem{$\WERE{}(S)$}{A graph $G=(V,E)$, three integers $k$, $r$ and $\lambda \leq r$, a weight function $\rho:V \cup E \rightarrow \{1,2,\dots\}$, a degree function $\delta: V \rightarrow \mathcal{P}([r])$, and a neighbourhood function $\nu: V\times V \rightarrow \mathcal{P}([\lambda])$.}{Can we obtain from $G$ a graph $G'=(V',E')$ using editing operations from $S$ only, such that for all $v \in V'$ we have $d^{\rho}(v) \in \delta(v)$ and for every $uv \in E'$ we have $\Card{N(u)\cap N(v)} \in \nu(u,v)$, with total editing cost at most $k$?}

Again we write $\WERE_{\plain}$ when we consider unweighted graphs, and $\WERE^{*}$ when $\delta$ and $\nu$ are restricted to singletons. If $\delta$ and $\nu$ are restricted to $\{r\}$ and $\{\lambda\}$ respectively we write $\WERE^{r,\lambda}$.

\subsubsection{Strong Regularity}\label{subsubsec:strong_ref_def}

$(r,\lambda,\mu)$-strongly regular graphs are $(r,\lambda)$-edge regular graphs where for all vertices $u,v \in V(G)$ such that $uv \notin E$ we have $\Card{N(u) \cap N(v)} = \mu$. For this set of constraints, our problem becomes the \textsc{Weighted Strongly Regular Editing} ($\WSRE$) problem. 

\problem{$\WSRE{}(S)$}{A graph $G=(V,E)$, four integers $k$, $r$ and $\lambda, \mu \leq r$, a weight function $\rho:V \cup E \rightarrow \{1,2,\dots\}$, a degree function $\delta: V \rightarrow \mathcal{P}([r])$, and two neighbourhood functions $\nu: V \times V  \rightarrow \mathcal{P}([\lambda])$ and $\xi: V \times V \rightarrow \mathcal{P}([\mu])$.}{Can we obtain from $G$ a graph $G'=(V',E')$ using editing operations from $S$ only, such that for all $v \in V'$ we have $d^{\rho}(v) \in \delta(v)$, for every $uv \in E'$ we have $\Card{N(u)\cap N(v)} \in \nu(u,v)$, and for every $uv \notin E'$ we have $\Card{N(u)\cap N(v)} \in \xi(u,v)$, with total editing cost at most $k$?}

We denote the case where $\delta$, $\nu$ and $\xi$ are restricted to singletons by $\WSRE^{*}$.

\section{Edge-Degree Regular Graphs}\label{sec:Edge-Degree_graphs}

Edge addition as a general operation is ill-defined for this probem, much as vertex addition makes little sense in the $\WDCE$ context, as the addition of a new edge requires the invention of constraints for that edge. Thus we restrict ourselves to vertex deletion and edge deletion. 

The $\NP{}$-completeness and $\W[1]$-hardness for parameter $k$ of $\WEDCE^{r}_{\plain}(S)$ where $\vdel{} \in S$ follow from the proof of Theorem~3.3 in~\cite{MathiesonSzeider12} as the graph constructed will be edge-degree $2r$-regular if and only if the same set of vertex deletions can be made as for the $r$-regular case.

When $S = \{\edel{}\}$ the $\NP{}$-completeness of $\WEDCE^{3}_{\plain}(\edel{})$ is established by the $\NP{}$-completeness of \textsc{Maximum H-Packing}~\cite{HellKirkpatrick78}, and more particularly by the \textsc{$K_{1,2}$-Packing} subproblem~\cite{PrietoSloper06}, obtained by setting $\delta(uv) = \{3\}$ for all edges $uv \in E(G)$.

\problem{$K_{1,2}$-Packing}{A graph $G$, an integer $k$.}{Does $G$ contain at least $k$ vertex-disjoint copies of the complete bipartite graph $K_{1,2}$?}

Of course this immediately gives the following:

\begin{proposition}
$\WEDCE^{r}_{\plain}(\edel{})$ is $\paraNP{}$-complete for parameter $r$.
\end{proposition}

It is a necessary condition for edge-degree $r$-regularity that the line graph is $r$-regular, however we cannot simply convert the graph to the line graph and perform the editing there, as the deletion of a vertex in the original graph does not have the same effect as the deletion of an edge in the line graph.

\subsection{A Bounded Search Tree Algorithm for WEDCE($\vdel{}, \edel{}$)}\label{subsec:BST_WEDCE}

Consider the case where both edge deletion and vertex deletion are allowed. Any isolated vertex can be discarded, as it will play no part in the edge-degree of any edge (this can be done without reducing $k$).

The algorithm, again based on Moser and Thilikos'~\cite{MT08}, runs as follows:

\begin{enumerate}
\item If $k \geq 0$ and for all edges $uv$, $d^{\rho}(uv) \in \delta(uv)$, answer \Yes{}. If $k \leq 0$ and there exists an edge $uv$ with $d^{\rho}(uv) \notin \delta(uv)$, answer \No{}.
\item Choose an edge $uv$ with $d^{\rho}(uv) \notin \delta(uv)$.
\item Arbitrarily select a set $M \subseteq N(u)\cup N(v)\setminus \{u,v\}$ of at most $r+1$ vertices. Let $E_M$ denote the edges with one endpoint in $M$ and the other as either $u$ or $v$.
\item Branch on all possibilities of deleting $u$, $v$, $uv$, one element $x$ of $M$ or reducing the weight of one edge in $E_M$ by $1$. Reduce $k$ by the weight of the deleted element in the first four cases, and by one if an edge is reduced in weight.
\item Return to step 1.
\end{enumerate}

The branching set consists of an edge $uv$ where $d^{\rho}(u)+d^{\rho}(v) \notin \delta(uv)$, $u$, $v$ at most $r+1$ neighbours of $u$ and $v$, and the edges between $u$ and $v$ and the $r+1$ neighbours. The branching set has at most $2r+5$ elements, thus the tree has at most $\tr(2r+5,k) = ((2r+5)^{k+1}-1)/(2r+4)$ vertices.

This gives the following result:

\begin{lemma}
$\WEDCE(S)$ is fixed-parameter tractable for parameter $(k,r)$ where $\emptyset \neq S \subseteq \{\vdel{},\edel{}\}$.
\end{lemma}

If only vertex deletion is allowed, the branching set may be reduced, resulting in a tree with at most $\tr(r+3,k)=((r+3)^{k+1}-1)/(r+2)$ vertices.

\subsection{A Kernelization for WEDCE$^{*}(\vdel{},\edel{})$}\label{subsec:kernel_WEDCE}

$\WEDCE^{*}(\vdel{},\edel{})$ admits a kernelization similar to $\WDCE(\vdel{},\edel{})$.

\subsubsection{Reduction Rules}\label{subsubsec:edge_degree_reduction_rules}

\begin{redrule}\label{redrule:large_degree}
  Let $(G,(k,r))$ be an instance of $\WEDCE^{*}(S)$.  If there is a vertex $v$
  in $G$ such that $d^\rho(v) > k+r$, then replace $(G,(k,r))$ with
  $(G',(k',r))$ where $G' = G - v$ and $k' = k - \rho(v)$.
\end{redrule}

Reduction Rule~$1$ for the $\WDCE^{\ast}$ problem~\cite{MathiesonSzeider12} states that given a vertex $v$ of degree greater than $k+r$, if there is a solution with at most $k$ deletions, $v$ must be deleted and $k$ reduced by $\rho(v)$. This rule also holds for $\WEDCE^{*}$.

\begin{claim}\label{claim:RR1_for_WEDCE}
  Reduction Rule~\ref{redrule:large_degree} is sound for $\WEDCE^{*}(S)$ with $\{\vdel{}\} \subseteq S \subseteq \{\vdel{},\edel{}\}$.
\end{claim}

\begin{proof}
  Assume there is a vertex $v \in V(G)$ with $d(v) > k+r$, then every edge $uv$ has edge-degree at least $k+r+1$ therefore at least $k+1$ vertices or edges must be deleted if we do not delete $v$, but we may only perform at most $k$ deletions.

  Thus $(G,(k,r))$ is a \Yes{}-instance of $\WEDCE^{*}(S)$ if and only if
  $(G',(k',r))$ is a \Yes{}-instance of $\WEDCE^{*}(S)$.
\end{proof}

We now adapt the notion of a clean region~\cite{MT08, MathiesonSzeider12} for $\WEDCE$. An edge $uv$ is clean if $\delta(uv) = d^{\rho}(uv)$. Let $X$ be the set of vertices only incident on clean edges. A clean region $C$ is a maximal connected subgraph of $G$ whose vertices are all in $X$ (or equivalently, a connected component of $G[X]$). We denote the vertices (resp. edges) of $C$ by $V(C)$ (resp. $E(C)$). As before let the $i$-th layer of $C$ be the subset $C_i=\{ c\in V(C) \mid \min_{b \in B} d_G(c,b) = i\}$ where $d_G(c,b)$ denotes the distance between $c$ and $b$ in $G$. Note that all the neighbors of a vertex of layer $C_i$ belong to $C_{i-1}\cup C_i \cup C_{i+1}$.

Reduction Rule~$2$ for $\WDCE^{\ast}$ holds immediately with this new definition of a clean region as independent clean regions will not have any elements deleted in a minimal solution.

\begin{redrule}\label{redrule:disjoint_clean}
  Let $(G,(k,r))$ be an instance of $\WDCE^{*}(S)$, let $C$ be a clean
  region of $G$ with empty boundary $B(C)=\emptyset$, and let $G' = G
  - V(C)$. Then replace $(G,(k,r))$ with $(G',(k,r))$.
\end{redrule}

If any element of a clean region is deleted, then the entire region must be, however with both $\vdel{}$ and $\edel{}$ it is not clear as to the most efficient way to delete a clean region. We do know that if any element from a layer with index greater than $k+1$ is deleted, either there is a solution where that element is not deleted, or there have been more than $k$ deletions.

\begin{redrule}\label{redrule:deep_clean_region}
Let $(G,(k,r))$ be an instance of $\WEDCE^{*}(S)$ and let $C$ be a clean region of $G$ such that $C_{k+2} \neq \emptyset$. Replace $(G,(k,r))$ with $(G',(k,r))$ as follows:
\begin{enumerate}
\item Delete all layers $C_{i}$ where $i \geq k+2$.
\item For each edge $uv$ such that $u \in C_{k+1}$ and $v \in C_{k+1} \cup C_{k}$ set $\delta(uv)$ to $d^{\rho}(uv)$.
\end{enumerate}
\end{redrule}

\begin{claim}\label{claim:deep_clean_for_WEDCE}
Reduction Rule~\ref{redrule:deep_clean_region} is sound for $\WEDCE^{*}(S)$ with $\emptyset \neq S \subseteq \{\vdel{}, \edel{}\}$.
\end{claim}

\begin{proof}
Let $D$ be the set of vertices and edges deleted in a minimal solution for $(G,(k,r))$. Let $G(D)$ be the subgraph induced by all the vertices of $D$ or incident to edges of $D$. Each connected component $X$ of $G(D)$ that contains an element of a clean region $C$ must also contain an element of the boundary $B(C)$, otherwise we could obtain a solution of lower cost by not deleting $X$. Therefore each vertex $v \in D \cap V(C)$ must be of distance at most $\Card{D}$ from a vertex in $B(C)$, thus there is no vertex $v \in C_{i}$ where $i \geq k+1$. Similarly any endpoint of an edge in $e \in D \cap E(C)$ must belong to some layer $C_{i}$ where $i \leq k+1$. Conversely any solution $D'$ for $(G',(k,r))$ is also a solution for $(G,(k,r))$
\end{proof}

If the editing operations are restricted to $\edel{}$, we can apply the following reduction rule. In this case we may contract the clean regions significantly. Note that here the weights of the vertices are irrelevant.

\begin{redrule}\label{redrule:clean_region_contraction}
Let $(G,(k,r))$ be an instance of $\WEDCE^{*}(S)$. Let $C$ be a clean region with boundary $B(C)$ and $\Card{V(C)} \geq 2$. We replace $(G,(k,r))$ with $(G'(k,r))$ by contracting $C$ to a single edge as follows:
\begin{enumerate}
\item Add two vertices $u$ and $v$ and the edge $uv$.
\item For each $b \in B(C)$ add an edge $bu$ with weight $\rho(bu) = \min(k+1,\sum_{c \in V(C)} \rho (bc))$.
\item For each $b \in B(C)$ set $\delta(bu) = d^{\rho}(b)+d^{\rho}(u)$.
\item Let $\rho(uv) = \min(k+1,\sum_{x,y\in V(C), xy \in C} \rho(xy))$.
\item Let $\delta(uv) = d^{\rho}(u)+d^{\rho}(v) = d^{\rho}(u)+1$.
\item Delete $C$.
\end{enumerate}
\end{redrule}

\begin{claim}\label{claim:RR7_for_WEDCE}
Reduction Rule~\ref{redrule:clean_region_contraction} is sound for $\WEDCE^{*}(\edel{})$.
\end{claim}

\begin{proof}
Let $(G,(k,r))$ be an instance of $\WEDCE^{*}(\edel{})$ with clean region $C$ having boundary $B(C)$. Let $(G',(k,r))$ be the instance obtained by applying Reduction Rule~\ref{redrule:clean_region_contraction}.

If an edge incident on a vertex of $V(C) \cup B(C)$ is deleted, then an edge of $C$ will no longer be clean. Therefore the edge must be deleted, similarly rendering other edges of $C$ no longer clean. This cascades, and all the edges of the clean region must be deleted with total cost equal to the sum of the cost of all the edges in the clean region. Therefore any solution for $(G,(k,r))$ either deletes all edges in $C$ or none. Therefore $C$ can be represented by a single edge of appropriate weight, which we can limit to $k+1$, as this weight or higher prevents deletion.
\end{proof}

\subsubsection{Kernelization Lemmas}\label{subsubsec:edge_degree_kernelization}

\begin{lemma}\label{lemma:edgedegree_kern1}
Let $\{\vdel{}\} \subseteq S \subseteq \{\vdel{},\edel{}\}$. Let $(G,(k,r))$ be a \Yes{}-instance of $\WEDCE^{*}(S)$ reduced under Reduction Rules~\ref{redrule:large_degree},~\ref{redrule:disjoint_clean} and~\ref{redrule:deep_clean_region}. Then $\Card{V(G)} \leq k(1+(k+r)(1+r^{k+1})) = O(k^2r^{k+1}+kr^{k+2})$. 
\end{lemma}

\begin{proof}
The proof runs identically to Lemma~$6.4$ in~\cite{MathiesonSzeider12}. 
\end{proof}

For $\WEDCE^{*}(\edel{})$ we can do much better.

\begin{lemma}\label{lemma:edgedegree_kern2}
Let $(G,(k,r))$ be a \Yes{}-instance of $\WEDCE^{*}(\edel{})$ reduced under Reduction Rules~\ref{redrule:large_degree},~\ref{redrule:disjoint_clean} and~\ref{redrule:clean_region_contraction}. Then $\Card{V(G)} \leq 2k + 4kr = O(kr)$. 
\end{lemma}

\begin{proof}
Let $D$ be the set of edges deleted in the solution. Let $H$ be the set of vertices incident to elements of $D$ and let $X$ be the remaining vertices of the graph . $\Card{D} \leq k$ by definition. As $\edel{}$ is the only operation, $D$ consists entirely of edges. Therefore $H \leq 2\cdot\Card{D}\leq 2k$.

\begin{claim}
$\Card{X} \leq 2r\cdot\Card{H}$.
\end{claim}

As $G-D$ is clean, the edges of $G-D$ incident to vertices of $H$ must have edge-degree at most $r$. Therefore the vertices of $H$ have degree at most $r$ in $G-D$ and all neighbours in $X$ must be vertices of clean regions. Furthermore Reduction Rule~$7$ any clean regions in the graph have at most $2$ vertices. The claim follows.

As $\Card{V(G)} = \Card{H}=\Card{X}$, $\Card{V(G)} \leq 2k + 4kr$.
\end{proof}

Combining these kernelizations, we obtain the following Theorem:

\begin{theorem}\label{thm:WEDCE_FPT_by_kern}
Let $\emptyset \neq S \subseteq \{\vdel{}, \edel{}\}$. $\WEDCE^{*}(S)$ is fixed-parameter tractable with a kernel with $O(k^2r^{k+1}+kr^{k+2})$ vertices. If $S=\{\edel{}\}$, the kernel has $O(kr)$ vertices.
\end{theorem}

\section{Edge Regular Graphs}\label{sec:edge_regular}

The $\W[1]$-hardness and $\NP{}$-completeness of $\WERE^{*}_{\plain}(S)$ and subsequently $\WERE(S)$ with $\{\vdel{}\} \neq S \subseteq \{\vdel{}, \edel{}, \eadd{}\}$ follow immediately from Theorem~$3.3$ in~\cite{MathiesonSzeider12}, as the solution required for the proof must be edge regular.

\subsection{A Bounded Search Tree Algorithm for WERE($\vdel{},\edel{}$)}\label{subsec:BST_WERE}

If any vertex $v$ has $d^{\rho} < h$ for all $h \in \delta(v)$ it can only be deleted, and $k$ reduced by $\rho(v)$. Then the algorithm is:

\begin{enumerate}
\item If $k\geq 0$ and for every vertex $v$ we have $d^{\rho}(v) \in \delta(v)$ and for every edge $uv$ we have $\Card{N(u)\cap N(v)} \in \nu(u,v)$ answer \Yes{}. If $k \leq 0$ and there exists a vertex $v$ with $d^{\rho}(v) \notin \delta(v)$ or an edge $uv$ with $\Card{N(u)\cap N(v)} \notin \nu(u,v)$ answer \No{}.
\item Choose a vertex $v$ with $d^{\rho}(v) \notin \delta(v)$, or incident to an edge $uv$ with $\Card{N(u)\cap N(v)} \notin \nu(u,v)$.
\item Arbitrarily select a set $M$ of at most $r+1$ vertices from $N(v)$. Let $E_{M}$ denote the edges incident on $v$ with the other endpoint in $M$.
\item Branch on all possibilities of deleting, $v$, one element $x$ of $M$ or reducing the weight of one edge in $E_{M}$ by $1$. Reduce $k$ by the weight of the deleted element in the first two cases, and by one if an edge is reduced in weight.
\item Return to step 1.
\end{enumerate}

The branching set for this problem is more complex. In both cases we choose an element from $V(G)\cup E(G)$. If the element chosen is a vertex $v$ with $d^{\rho}(v) \notin \delta(v)$, it consists of $v$, at most $r+1$ neighbours of $v$ and the edges between $v$ and the chosen neighbours. If the element chosen is an edge $uv$ with $\Card{N(u)\cap N(v)} \notin \nu(u,v)$, it consists of $uv$, $u$, $v$, at most $r+1$ neighbours of $u$ and $v$, and the edges between $u$ and $v$ and the chosen neighbours. Therefore the tree has at most $\tr(3r+6,k)=((3r+6)^{k+1}-1)/(3r+5)$ vertices. As with previous cases, restricting the available operations to $\vdel{}$ gives a smaller bound on the tree size. With only $\vdel{}$ the branching set is reduced to at most $r+3$ elements, so the maximum number of vertices in the tree is reduced to $\tr(r+3,k)=((r+3)^{k+1}-1)/(r+2)$.

\begin{lemma}
$\WERE(S)$ is fixed-parameter tractable with paramter $(k,r)$ with $\emptyset \neq S \subseteq \{\vdel{}, \edel{}\}$.
\end{lemma}

\subsection{A Kernelization for WERE$^{*}(\vdel{}, \edel{})$}\label{subsec:kernel_WERE}

As before the $\WDCE^{*}(\vdel{}, \edel{})$ kernelization can be adapted for $\WERE^{*}(\vdel{}, \edel{})$ using an adapted definition of a clean region. As we are dealing with singleton sets for the degree function $\delta$ (resp. the neighbourhood function $\nu$) we write $d = \delta(v)$ (resp. $d = \nu(u,v)$) instead of $d \in \delta(v)$ (resp. $d \in \nu(u,v)$).

\subsubsection{Reduction Rules}

As an edge regular graph is by definition regular, Reduction Rule ~\ref{redrule:large_degree} (\emph{q.v.} Section~\ref{subsubsec:edge_degree_reduction_rules}) applies immediately.

We redefine a clean region as a maximal connected set $C\subseteq V(G)$ of vertices such that for every vertex $v \in C$ we have $d^{\rho}(v)=\delta(v)$ and for every edge $uv \in E(G)$ incident to $v$ we have $\Card{N(u)\cap N(v)} = \nu(u,v)$. 
The layers of $C$ are defined as in Section~\ref{subsubsec:edge_degree_reduction_rules}. We note that it follows immediately that for $u \in B(C)$ and $v \in C$ if $uv \in E$ then $N(u) \cap N(v) \subseteq B(C) \cup C_{1}$.

With this modified definition of a clean region, Reduction Rule~\ref{redrule:disjoint_clean} (\emph{q.v.} Section~\ref{subsubsec:edge_degree_reduction_rules}) also applies with no change.

We now only require an appropriate alternative for Reduction Rule~\ref{redrule:clean_region_contraction}.

\begin{redrule}\label{redrule:edge_regular_contraction}
Let $(G,(k,r))$ be an instance of $\WERE^{*}(S)$. Let $C$ be a clean region with boundary $B(C)$. We replace $(G,(k,r))$ with $(G',(k,r))$ by shrinking $C$ to a single layer of vertices as follows:

\begin{enumerate}
\item Delete all vertices $v \in C\setminus C_{1}$.
\item Add all edges $uv$ where $u,v \in C_{1}$.
\item For all edges $uv$ where $u,v \in C_{1}$, set $\nu(u,v)$ to $\Card{N(u)\cap N(v)}$.
\item For all vertices $v \in C_{1}$, set $\delta(v)$ to $d^{\rho}(v)$.
\item Arbitrarily choose a vertex $v \in C_{1}$ and set $\rho(v)$ to $\min(k+1, \rho(v)+\sum_{u \in C\setminus C_{1}} \rho(u))$.
\end{enumerate}
\end{redrule}

\begin{claim}\label{claim:RR8_for_WERE}
Reduction Rule~\ref{redrule:edge_regular_contraction} is sound for $\WERE^{*}(S)$ with $\{\vdel{}\} \subseteq S \subseteq \{\vdel{}, \edel{}\}$.
\end{claim}

\begin{proof}
As with Reduction Rule~\ref{redrule:clean_region_contraction}, if a neighbour, incident edge, or vertex of clean region is deleted, the entire clean region must be deleted. As the replacement clean region is connected, $\nu$ is satisfied for all vertices in $C$, and the weight is the same or $k+1$, this property and the cost of deletion is preserved.

As all vertices in $B(C)$ have their neighbours in the clean region confined to $C_{1}$, the retention of $C_{1}$ ensures that the neighbourhoods with regard to $\nu$ are also preserved.
\end{proof}

\subsubsection{Kernelization Lemma}

\begin{lemma}\label{lem:were_kernel}
Let  $\{\vdel{}\} \subseteq S \subseteq \{\vdel{}, \edel{}\}$. Let $(G,(k,r))$ be a \Yes{}-instance of $\WERE^{*}(S)$ reduced under Reduction Rules~\ref{redrule:large_degree},~\ref{redrule:disjoint_clean} and~\ref{redrule:edge_regular_contraction}. Then $\Card{V(G)} \leq k + k(k+r) + kr(k+r) = O(kr(k+r))$. 
\end{lemma}

\begin{proof}
We partition $G$ into disjoint sets $D$, $H$ and $X$ where $D$ is the set of vertices and edges deleted to obtain the edge regular graph, $H$ is the set of vertices incident or adjacent to elements of $D$ that are not themselves in $D$ and $X$ is the set of remaining vertices in neither $D$ nor $H$. $H$ separates $D$ from $X$. By definition $\Card{D} \leq k$.

We also have $\Card{H} \leq \Card{D}\cdot(k+r)$.
\begin{claim}\label{claim:were_X_size}
$\Card{X} \leq r\cdot\Card{H}$.
\end{claim}

As removing $D$ leaves the graph edge regular, the vertices of $H$ can have at most $r$ neighbours outside of $D$, in particular they can have at most $r$ neighbours in $X$. Furthermore all vertices in $X$ must belong to some clean region in the original graph, therefore by Reduction Rules~$2$ and $8$, all the vertices in $X$ must be adjacent to some vertex in $H$. The claim follows.

We have $\Card{V(G)} \leq \Card{D} + \Card{H} + \Card{X} \leq k + k(k+r) + kr(k+r)$.
\end{proof}

Subsequently we have the following:

\begin{lemma}\label{lem:WERE_FPT_by_kern}
$\WERE^{*}(S)$ is fixed parameter-tractable with parameter $(k,r)$ where $\{\vdel{}\} \subseteq S \subseteq \{\vdel{}, \edel{}\}$.
\end{lemma}

As we only have deletion operations, we can represent an unweighted graph as a weighted graph where all elements have weight $1$.

\begin{corollary}\label{cor:WERE_var_FPT_by_kern}
$\WERE^{r,\lambda}(S)$, $\WERE^{*}_{\plain}(S)$ and $\WERE^{r,\lambda}_{\plain}(S)$ are fixed parameter-tractable with parameter $(k,r)$ where $\{\vdel{}\} \subseteq S \subseteq \{\vdel{}, \edel{}\}$.
\end{corollary}

\section{Strongly Regular Graphs}\label{sec:strongly_regular}

\subsection{A Kernelization for WSRE$^{*}(\vdel{}, \edel{}$)}\label{subsec:kernel_WSRE}

\subsubsection{Reduction Rules}

Again Reduction Rule~\ref{redrule:large_degree} (\emph{q.v.} Section~\ref{subsubsec:edge_degree_reduction_rules}) holds with no change, and once the appropriate notion of clean region has been defined, Reduction Rule~\ref{redrule:disjoint_clean} also holds.

In this case the definition of a clean region is the obvious extension of the clean region for $\WERE{}$, where for each vertex $v$ in the clean $C$ region and every vertex $u$ where $uv \notin E$ we have $\Card{N(u) \cap N(v)} = \xi(u,v)$.

\begin{redrule}\label{redrule:strongly_regular_contraction}
Let $(G,(k,r))$ be an instance of $\WSRE^{*}(S)$. Let $C$ be a clean region with boundary $B(C)$. We replace $(G,(k,r))$ with $(G',(k,r))$ by shrinking $C$ by removing all but $C_{1}$ and $C_{2}$ as follows:
\begin{enumerate}
\item Delete $C\setminus(C_{1}\cup C_{2})$.
\item Add all edges $uv$ where $u,v \in C_{1}\cup C_{2}$.
\item For every vertex $v \in C_{2}$ set $\delta(v)$ to $d^{\rho}(v)$.
\item For every vertex $v \in C_{1} \cup C_{2}$ and every vertex $u \in V(G)$, if $uv \in E$ set $\nu(u,v)$ to $\Card{N(u) \cap N(v)}$, if $uv \notin E$ set $\xi(u,v)$ to $\Card{N(u) \cap N(v)}$.
\item Arbitrarily choose a vertex $v \in C_{1} \cup C_{2}$ and set $\rho(v)$ to $\min(k+1, \sum_{u\in C\setminus(C_{1}\cup C_{2})}\rho(u))$.
\end{enumerate}
\end{redrule}

\begin{claim}\label{claim:RR9_for_WSRE}
Reduction Rule~\ref{redrule:strongly_regular_contraction} is sound for $\WSRE^{*}(S)$ with $\{\vdel{}\} \subseteq S \subseteq \{\vdel{}, \edel{}\}$.
\end{claim}

\begin{proof}
The new clean region remains a clean region, and the total weight, if less than $k+1$ is the same. Any total weight greater than $k$ is equivalent, as we cannot delete the clean region in any of those cases.

As $C$ is a clean region, if there is a vertex $v$ with $xi(u,v) > 0$ for some vertex $u \in C$, then $d(u,v) \leq 2$. Therefore we do not affect vertices outside of the clean region by removing all layers $C_{i}$ for $i \geq 3$.
\end{proof}

\subsubsection{Kernelization Lemma}

\begin{lemma}\label{lem:WSRE_kern}
Let  $\{\vdel{}\} \subseteq S \subseteq \{\vdel{}, \edel{}\}$. Let $(G,(k,r))$ be a \Yes{}-instance of $\WSRE^{*}(S)$ reduced under Reduction Rules~\ref{redrule:large_degree},~\ref{redrule:disjoint_clean} and~\ref{redrule:strongly_regular_contraction}. Then $\Card{V(G)} \leq k + k(k+r) + kr(r+1)(k+r) = O(kr^{2}(k+r))$. 
\end{lemma}

\begin{proof}
We partition $G$ into disjoint sets $D$, $H$ and $X$ where $D$ is the set of vertices and edges deleted to obtain the edge regular graph, $H$ is the set of vertices incident or adjacent to elements of $D$ that are not themselves in $D$ and $X$ is the set of remaining vertices in neither $D$ nor $H$. $H$ separates $D$ from $X$. By definition $\Card{D} \leq k$.

Further, $\Card{H} \leq \Card{D}\cdot(k+r)$.

\begin{claim}\label{claim:wsre_X_size}
$\Card{X} \leq r(r+1)\cdot\Card{H}$.
\end{claim}

As removing $D$ leaves the graph edge regular, the vertices of $H$ can have at most $r$ neighbours outside of $D$, in particular they can have at most $r$ neighbours in $X$. Furthermore all vertices in $X$ must belong to some clean region in the original graph, therefore by Reduction Rules~$2$ and $9$, all the vertices in $X$ must be at most distance~$2$ from some vertex in $H$. Therefore the first layer $C_{1}$ of a clean region $C$ is of size at most $r\cdot\Card{H}$, and the second is of size at most $r^{2}\cdot\Card{H}$. The claim then follows.

We have $\Card{V(G)} \leq \Card{D} + \Card{H} + \Card{X} \leq k + k(k+r) + kr(r+1)(k+r)$.
\end{proof}

This gives the following:

\begin{lemma}\label{lem:WSRE_FPT_by_kern}
$\WSRE^{*}(S)$ is fixed-parameter tractable with parameter $(k,r)$ where $\{\vdel{}\} \subseteq S \{\vdel{}, \edel{}\}$.
\end{lemma}

As before the employment of only deletion operations allows the representation of an unweighted graph as a weighted graph.

\begin{corollary}\label{cor:WSRE_var_FPT_by_kern}
$\WSRE^{r,\lambda,\mu}(S)$, $\WSRE^{*}_{\plain}(S)$ and $\WSRE^{r,\lambda,\mu}_{\plain}(S)$ are fixed parameter-tractable with parameter $(k,r)$ where $\{\vdel{}\} \subseteq S \subseteq \{\vdel{}, \edel{}\}$.
\end{corollary}

\subsection{Fixed-Parameter Tractability for WSRE and WERE with Edge Addition}\label{subsec:WSRE_logic}

If we consider the case of $\WSRE(S)$ where $\{\vdel{},\eadd{}\} \subseteq S$, and for every vertex $v$ we have $\nu(v) = \xi(v) = \{0, \ldots, r\}$, then the problem is reduced to $\WDCE(S)$. By the $\NP{}$-completeness and compositionality of \textsc{Edge Replacement Set} (\emph{q.v.}~\cite{Mathieson10,MathiesonSzeider12}) we would expect that in general $\WSRE^{*}(S)$ and subsequently $\WSRE(S)$ would have no polynomial time kernelization that relied upon choosing vertices from clean regions to delete to obtain edge addition points, and that $\WSRE^{*}(S)$ and subsequently $\WSRE(S)$ would have no polynomial sized kernel unless the Polynomial Hierarchy collapses to the third level, by Lemma~\ref{lemma:comp_no_poly_kernel}. The same arguments apply to $\WERE$.

However we can still apply the logic approach used in~\cite{MathiesonSzeider12} to demonstrate fixed-parameter tractability for $\WSRE(S)$ and $\WERE(S)$ with $\emptyset \neq S \subseteq \{\vdel{}, \edel{}, \eadd{}\}$. In fact as we still require the graphs to satisfy the basic degree constraint $\delta$, we can extend the existing formul\ae{} from~\cite{MathiesonSzeider12}. We represent a graph by its incidence structure, with the additional relations $W_{i}$ and $D_{i}$ as in~\cite{MathiesonSzeider12}. To the vocubulary we add a further two sets of binary relations $N_{i}$ with $0 \leq i \leq r$ and $M_{j}$ with $0 \leq j \leq r$ which encode the functions $\nu$ and $\xi$ respectively, i.e., for a pair of vertices $u,v \in V(G)$ $N_{i}uv$ is true if and only if $i \in \nu(uv)$, similarly
$M_{i}uv$ is true if and only if $i \in \xi(uv)$. By setting $\xi(u,v) = \{0, \ldots, r\}$ for all vertices $u,v$ we can represent $\WERE$ with no change to the logic (in fact we may represent another otherwise undefined problem where $\nu(v) = \{0, \ldots, r\}$ but $\xi(u,v)$ does not).

To each formula $\phi_{k}$ as defined in~\cite{MathiesonSzeider12}, we add two subformul\ae{} $\forall uv\psi'_{k}$ and $\forall uv\psi''_{k}$ by conjunction.

First we repeat the construction of $\phi_{k}$, presented as in the previous work~\cite{MathiesonSzeider12}:


\[
  \phi_k=\bigvee_{\begin{array}{l}k',k'',k''' \in [k] \text{ such that }\\k'+k''+k''' \leq k\end{array}}\left[\begin{aligned}
\exists u_{1}, \ldots, u_{k'}, & e_{1}, \ldots, e_{k''},\\
\exists a_{1},\ldots, a_{k'''},&  b_{1},\ldots, b_{k'''}\\
&  (\phi_k' \wedge \forall v \ \phi_k'')]
\end{aligned}\right]\]

where $\phi_k'$ and $\phi_k''$ are given below. The subformula
$\phi_k'$ is the conjunction of the clauses (1)--(3) and ensures that
$u_{1}, \ldots, u_{k'}$ represent deleted vertices, $e_{1}, \ldots
e_{k''}$ represent deleted edges, $a_{i},b_i$, $1\leq i \leq k'''$
represent end points of added edges, and the total editing cost is at
most~$k$. Note that since added edges are not present in the given
structure we need to express them in terms vertex pairs.  For the
unweighted case we must also include subformul\ae{} (4) and~(5) to ensure
that the addition of edges does not produce parallel edges.  By
restricting $k'$, $k''$ or $k'''$ to zero as appropriate we can
express which editing operations are available.

\begin{enumerate}
\item[(1)] $\bigwedge_{i\in [k']} Vu_i \wedge \bigwedge_{i\in [k'']}
  Ee_i$
\ ``$u_i$ is a vertex, $e_i$ is an edge;''
\item[(2)] $\bigwedge_{i\in [k''']} Va_i \wedge Vb_i \wedge  a_i\neq
  b_i  \wedge \bigwedge_{j \in [k']} (u_j\neq a_i \wedge u_j\neq b_i)$
\  ``$a_i$ and $b_i$ are distinct vertices and  not deleted;''
\item[(3)]$\bigvee_{w_1,\dots,w_{k'}\in [k'] \text{ such that } 
 \sum_{i \in [k']}w_{i} + k''+k''' \leq  k } \bigwedge_{i \in [k']}
W_{w_{i}}u_{i}$ 
\ ``the weight of deleted vertices is correct;''
\item[(4)] $\bigwedge_{1\leq i < j \leq k''}
(a_i\neq b_j \vee a_j \neq b_i)
\wedge
(a_i\neq a_j \vee b_i \neq b_j)$
\ ``the pairs of vertices are mutually distinct;''
\item[(5)] $\bigwedge_{i\in [k''']} \forall y (\neg Ia_iy \vee \neg
  Ib_iy)$ 
\ ``$a_i$ and $b_i$ are not adjacent.''
\end{enumerate}

The subformula $\phi_k''$ ensures that after editing each vertex $v$
has degree $l \in \delta(v)$.
\[\phi_k''=
( Vv  \wedge \bigwedge_{i \in [k']} v
\neq u_{i}) \rightarrow
\bigvee_{l \in [r]}D_{l}v \wedge \bigvee_{\scriptsize
  \begin{array}{c}
    l',l''\in[l]\\
    l'+l''=l
  \end{array}
} \exists x_1,\dots,x_{l'},
 y_1,\dots,y_{l''}
\ \phi_k''',
\]

where $\phi_k'''$ is the conjunction of the clauses (6)--(12).
\begin{enumerate}
\item[(6)] $\bigwedge_{i \in [l']} Ivx_{i}$
\ ``$v$ is incident with $l'$ edges;''
\item[(7)] $\bigwedge_{1\leq i<j \leq l'} x_{i} \neq x_{j}$
\ ``the edges are all different;''
\item[(8)] $\bigwedge_{i \in [l'], j \in [k'']} x_{i} \neq e_{j}$ 
\ ``the  edges have not been deleted;''
\item[(9)] $\bigwedge_{i\in[l'],j\in[k']} \neg Iu_jx_{i}$
\ ``the ends of the edges have not been  deleted;''
\item[(10)] $\forall x (Ivx \rightarrow \bigvee_{i \in [l']} x = x_{i}
  \vee \bigvee_{i \in [k'']} x = e_{i} \vee \bigvee_{i} Ixu_i)$
\  ``$v$ is not incident with any further edges except deleted edges;''
\item[(11)] $\bigwedge_{i \in [l'']} \bigvee_{j \in [k''']}
  (y_i=a_j\wedge v=b_j) \vee (y_i=b_j \wedge v=a_j)$ 
 \ ``$v$ is incident with at least $l''$ added edges;''
\item[(12)] $\bigwedge_{j\in [l'']} (v=a_j \rightarrow
  \bigvee_{i} y_i=b_j) \wedge (v=b_j \rightarrow \bigvee_{j\in[l'']}
  y_i=a_j)$ 
\ ``$v$ is incident with at most $l''$ added edges.''
\end{enumerate}

To this we include further subformul\ae{} to accomodate the edge constraints.
The subformula $\psi'_{k}$ ensures that if two vertices $u$ and $v$ are adjacent, then $\nu(u,v)$ is satisfied. 

\begin{align*}&\psi'_{k} = (Vu \wedge Vv \wedge \bigwedge_{i\in [k']} (u_{i} \neq u \wedge u_{i} \neq v)\\
  &\wedge \exists e (Iue \wedge Ive \wedge \bigwedge_{i \in [k'']} (e_{i} \neq e))) \rightarrow\\ 
  &\bigvee_{l \in [r]} N_{l}uv \wedge\\
  &\bigvee_{\scriptsize\begin{array}{c}l',m',l'',m'' \in [l]\\l'+m' =l''+m'' = l\end{array}} \exists n_{1}\ldots n_{l}x_{1}\ldots x_{l'}y_{1}\ldots y_{m'}w_{1}\ldots w_{l''}z_{1}\ldots z_{m''} \psi'''_{k}
\end{align*}

Where $\psi'''_{k}$ is the conjuction of clauses (13)--(30):

\begin{enumerate}
\item[(13)] $\bigwedge_{i \neq j \in [l]} Vn_{i} \wedge n_{i} \neq n_{j} \wedge u \neq n_{i} \wedge v \neq n_{j}$ \ ``the $n_{i}$s are distinct vertices, different from $u$ and $v$;''
\item[(14)] $\bigwedge_{i \in [l],j\in [k']} n_{i} \neq u_{j}$ \ ``the $n_{i}$s have not been deleted;''
\item[(15)] $\bigwedge_{i \neq j \in [l']} Ex_{i} \wedge x_{i} \neq x_{j}$ \ ``the $x_{i}$s are distinct edges;''
\item[(16)] $\bigwedge_{i \in [l'], j \in [k'']} x_{i} \neq e_{j}$ \ ``the $x_{i}$s have not been deleted;''
\item[(17)] $\bigwedge_{i \neq j \in [m']} Vy_{i} \wedge y_{i} \neq y_{j}$ \ ``the $y_{i}$s are distinct vertices;''
\item[(18)] $\bigwedge_{i \in [m'], j \in [k']} y_{i} \neq u_{j}$ \ ``the $y_{i}$s have not been deleted;''
\item[(19)] $\bigwedge_{i \in [m']} \bigvee_{j \in [l]} y_{i} = n_{j}$ \ ``the $y_{i}$s are equal to some $n_{j}$;'' 
\item[(20)] $\bigwedge_{i \in [l']} (Iux_{i} \wedge \bigvee_{j \in [l]} In_{j}x_{i}\wedge \bigwedge_{t \in [m']} n_{j} \neq y_{t})$ \ ``the $x_{i}$s are adjacent to $u$ and some $n_{j}$ which is not equal to any $y_{t}$;''
\item[(21)] $\bigwedge_{i \in [m']}\bigvee_{j \in [k''']} (u = a_{j} \wedge y_{i} = b_{j}) \vee (u = b_{j} \wedge y_{i} = a_{j})$ \ ``$y_{i}$ and $u$ are the endpoints of some added edge;''
\item[(22)] $\bigwedge_{i \in [k''']} a_{i} = u \vee b_{i} = u \rightarrow \bigvee_{j \in [m']} a_{i} = y_{j} \vee b_{i} = y_{j}$ \ ``$u$ is incident on no other added edges.''
\item[(23)] $\bigwedge_{i \neq j \in [l'']} Ew_{i} \wedge w_{i} \neq w_{j}$ \ ``the $w_{i}$s are distinct edges;''
\item[(24)] $\bigwedge_{i \in [l''], j \in [k'']} w_{i} \neq e_{j}$ \ ``the $w_{i}$s have not been deleted;''
\item[(25)] $\bigwedge_{i \neq j \in [m'']} Vz_{i} \wedge z_{i} \neq z_{j}$ \ ``the $z_{i}$s are distinct vertices;''
\item[(26)] $\bigwedge_{i \in [m''], j \in [k']} z_{i} \neq u_{j}$ \ ``the $z_{i}$s have not been deleted;''
\item[(27)] $\bigwedge_{i \in [m'']} \bigvee_{j \in [l]} z_{i} = n_{j}$ \ ``the $z_{i}$s are equal to some $n_{j}$;'' 
\item[(28)] $\bigwedge_{i \in [l'']} (Iuw_{i} \wedge \bigvee_{j \in [l]} In_{j}w_{i}\wedge \bigwedge_{t \in [m'']} n_{j} \neq z_{t})$ \ ``the $w_{i}$s are adjacent to $u$ and some $n_{j}$ which is not equal to any $z_{t}$;''
\item[(29)] $\bigwedge_{i \in [m'']}\bigvee_{j \in [k''']} (u = a_{j} \wedge z_{i} = b_{j}) \vee (u = b_{j} \wedge z_{i} = a_{j})$ \ ``$z_{i}$ and $u$ are the endpoints of some added edge;''
\item[(30)] $\bigwedge_{i \in [k''']} a_{i} = u \vee b_{i} = u \rightarrow \bigvee_{j \in [m'']} a_{i} = z_{j} \vee b_{i} = z_{j}$ \ ``$u$ is incident on no other added edges.''
\end{enumerate}

The subformula $\psi''_{k}$ ensures that for two nonadjacent vertices $u$ and $v$, $\xi(u,v)$ is satisfied. $\psi''_{k}$ is essentially identical to $\psi'_{k}$, and we can re-use the subformula $\psi'''_{k}$.

\begin{align*}
  &\psi''_{k} = (Vu \wedge Vv \wedge \bigwedge_{i\in [k']} (u_{i} \neq u \wedge u_{i} \neq v) \wedge\\
  &\forall e (\neg Iue \vee \neg Ive \vee \bigwedge_{i \in [k'']} (e_{i} \neq e))) \rightarrow\\
  &\bigvee_{l \in [r]} M_{l}uv \wedge\\
  &bigvee_{\scriptsize\begin{array}{c}l',m',l'',m'' \in [l]\\l'+m' =l''+m'' = l\end{array}} \exists n_{1}\ldots n_{l}x_{1}\ldots x_{l'}y_{1}\ldots y_{m'}w_{1}\ldots w_{l''}z_{1}\ldots z_{m''} \psi'''_{k}
\end{align*}

\begin{lemma}\label{lem:logic_formula_exist_for_WSRE}
Let $(G,(k,r))$ be an instance of $\WSRE(S)$ (resp. $\WERE(S)$) where $\emptyset \neq S \subseteq \{\vdel{}, \edel{}, \eadd{}\}$ with associated incidence structure $S_{G}$. There exist first order formul\ae{} $\phi_{k}$, for $k \geq 0$, such that $S_{G}$ is a model for $\phi_{k}$ if and only if $(G,(k,r))$ is a \Yes{}-instance of $\WSRE(S)$ (resp. $\WERE(S)$).
\end{lemma}

By Frick and Grohe's metatheorem~\cite{FG01} and Lemma~\ref{lem:logic_formula_exist_for_WSRE}:

\begin{theorem}\label{thm:WSRE_WERE_FPT_by_logic}
The problems $\WSRE(S)$ and $\WERE(S)$ are fixed-parameter tractable for parameter $(k,r)$ where $\emptyset \neq S \subseteq \{\vdel{}, \edel{}, \eadd{}\}$.
\end{theorem}

\section{WDCE and Treewidth}\label{sec:WDCE_tw}

We now return to the $\WDCE$ problem with an alternate parameterization, the treewidth $\tw(G)$ of the input graph $G$. There are several options for parameterizing, dependent on what combination of the treewidth, the degree bound $r$ and the editing cost $k$ is chosen. Of course if both $k$ and $r$ are part of the parameterization, we already have a complete classification in Theorem~\ref{thm:queen}. It can also be observed that if a graph $G$ has treewidth $\tw(G) \leq t$, then there is some vertex $v \in V(G)$ with $d(v) \leq t$. Therefore for $\WDCE^{r}_{\plain}(\vdel{},\edel{})$ if $r > t$ we may immediately answer \No{}. The vertex $v$ with $d(v) \leq t$ must be deleted, however the resultant graph $G'$ has $\tw(G') \leq t$, therefore this process cascades and the entire graph must be deleted. Furthermore if $r \leq t$ and $k$ is also a parameter, then Theorem~\ref{thm:queen} applies. As $\WDCE^{*}(S)$ for $\emptyset \neq S \subseteq \{\edel{},\eadd{}\}$ is in $\P{}$, parameterization by any combination of treewidth, $k$ and $r$ does not affect the complexity. 

Samer and Szeider~\cite{SS08} show that the \textsc{General Factor}${}=\iWDCE_{\plain}(\edel{})$ problem is $\W[1]$-hard when parameterized by treewidth alone. We can extend this to $\WDCE(\edel{})$.

\begin{proposition}
$\WDCE_{\plain}(\edel{})$ is $\W[1]$-hard when parameterized by the treewidth of the input graph. Furthermore it remains $\W[1]$-hard when the input graphs are bipartite and the degree constraints of one partite set are limited to $\{1\}$.
\end{proposition}

\begin{proof}
Samer and Szeider's proof~\cite{SS08} establishes the $\W[1]$-hardness of $\iWDCE_{\plain}(\edel{})$ when restricted to bipartite graphs where the degree constraints of one partite set are $\{1\}$. We simply choose $k = \sum_{uv \in E(G)} \rho(uv)$, where $G$ is the input graph.
\end{proof}



If we set the vertex weights appropriately, we can also allow vertex deletion. However we can no longer claim unit weights.

\begin{corollary}
$\WDCE(\vdel{},\edel{})$ is $\W[1]$-hard when parameterized by the treewidth of the input graph. Furthermore it remains $\W[1]$-hard when the input graphs are bipartite and the degree constraints of one partite set are limited to $\{1\}$.
\end{corollary}

\begin{proof}
For all vertices $v \in V(G)$ set $\rho(v) = k+1$. Then no vertex can be deleted within the cost.
\end{proof}



By subdividing the edges and weighting the original vertices we can restrict the operations to vertex deletion alone.

\begin{corollary}
$\WDCE(\vdel{})$ is $\W[1]$-hard when parameterized by the treewidth of the input graph. Furthermore it remains $\W[1]$-hard when the input graphs are bipartite and the degree constraints of one partite set are limited to $\{1\}$.
\end{corollary}

\subsection{Parameterizations Excluding $k$}\label{subsec:WDCE_tw_no_k}

If we consider version of the problem where the number of edit operations is unbounded, we can obtain some further results for limited cases. In this setting, as the number of deletions is unbounded, we do not consider the trivial case where $V(G) = \emptyset$ as a valid solution.

\begin{lemma}
$\iWDCE^{r}_{\plain}(\vdel{})$ is fixed-parameter tractable when parameterized by the treewidth of the input graph.
\end{lemma}

\begin{proof}
As noted earlier, if $r > \tw(G)$ for a graph $G$, then $(G,\tw(G))$ is a \No{}-instance of $\iWDCE^{r}_{\plain}(\vdel{})$, as the entire graph would have to be deleted. However if $r \leq \tw(G)$, we may apply Courcelle's Theorem with the following second order sentence:
\[
\begin{array}{r}
\exists S \forall v \forall u(Vv \rightarrow Sv \vee \exists v_{1}, \ldots, v_{r}(\bigwedge_{i\neq j \in [r]}(v_{i} \neq v_{j}) \wedge\\ \bigwedge_{i \in [r]} (\neg Sv_{i} \wedge Avv_{i} \wedge v \neq v_{i}) \wedge\\ (Avu \rightarrow Su \vee \bigvee_{i \in [r]} u = v_{i})))
\end{array}
\]
where $Axy$ is shorthand for $\exists e(Ee\wedge Vx \wedge Vy \wedge Ixe \wedge Iye)$ (i.e, $x$ and $y$ are adjacent). The sentence ensures that there is a set $S$ (the deleted vertices) such that for every vertex $v$ and every vertex $u$, either $v$ is deleted, or it is adjacent to $r$ distinct vertices that haven't been deleted, and if $u$ is adjacent to $v$, then it is one of these vertices, or it has been deleted.
\end{proof}

This can be extended to include edge deletion.

\begin{lemma}
$\iWDCE^{r}_{\plain}(\vdel{},\edel{})$ is fixed-parameter tractable when parameterized by the treewidth of the input graph.
\end{lemma}

\begin{proof}
As before if $r > \tw(G)$, the instance is a \No{}-instance. Then we need only construct a second order logic sentence that encodes the problem.
\[
\exists S \forall v \forall e(Vv \rightarrow Sv \vee (\exists e_{1},\ldots,e_{r}, v_{1},\ldots,v_{r}(\phi_{1}\wedge\phi_{2})))
\]
where $\phi_{1}$ is the conjunction of subclauses (1)--(5):
\begin{enumerate}
\item[(1)] $\bigwedge_{i \in [r]} \neg Se_{i} \wedge \neg Sv_{i}$\ ``$e_{i}$ and $v_{i}$ have not been deleted;''
\item[(2)] $\bigwedge_{i \in [r]} Ee_{i} \wedge Vv_{i}$\ ``$e_{i}$ is an edge and $v_{i}$ is a vertex;''
\item[(3)] $\bigwedge_{i \in [r]} v_{i} \neq v$\ ``$v$ is not equal to any $v_{i}$;''
\item[(4)] $\bigwedge_{i \in [r]} Iv_{i}e_{i} \wedge Ive_{i}$\ ``$v$ and $v_{i}$ are adjacent;''
\item[(5)] $\bigwedge_{i \neq j \in [r]} v_{i} \neq v_{j}$\ ``the $v_{i}$s are distinct.''
\end{enumerate}
and 
\[
\phi_{2} = Ive \rightarrow (\bigvee_{i \in [r]} (e = e_{i}) \vee Se \vee \exists u(Iue \wedge u \neq v \wedge Su))
\]
$\phi_{2}$ ensures that if there is an edge incident to $v$, then either it is one of the $r$ edges making up the the regular degree of $v$, it was deleted, or its other endpoint was deleted.
\end{proof}

If vertex deletion and edge addition are allowed, then the problem becomes trivially polynomial.

\begin{lemma}
$\iWDCE^{r}_{\plain}(\vdel{},\edel{},\eadd{})$ and $\iWDCE^{r}_{\plain}(\vdel{}, \eadd{})$ are polynomial-time solvable.
\end{lemma}

\begin{proof}
As the number of edit operations is unlimited, we can simply delete all but $r+1$ vertices, and make the graph a $K_{r+1}$.

If there are less than $r+1$, it is not possible to have an $r$-regular graph, and we answer \No{} immediately.
\end{proof}

\section{Conclusion}
\label{sec:conclusion}

We have examined a series of editing problems with constraints based on natural extensions of regularity. As with the \WDCE{} series of problems, problems with extended edge-degree based constraints are in \FPT{} with combined parameter $k+r$, but \W[1]-hard with parameter $k$ and \paraNP{}-complete with paramter $r$. There are a number of avenues of further research open with these problems, the mostprominent being the development of a concrete algorithm for $\WSRE{}(\vdel{}, \edel{}, \eadd{})$. The techniques of~\cite{Golovach15} may be applicable, however this is certainly not trivial.

We also consider parameterizations of variants of the \WDCE{} problem by the treewidth of the input graph and demonstrate tractability for these problems.

\bibliographystyle{plain}
\bibliography{final}

\end{document}